\documentclass[11pt,fullpage]{article}
\usepackage{graphicx} 
\usepackage{amsmath}
\usepackage{natbib}
\def\withcolors{1}
\def\withnotes{1}
\usepackage{ccanonne}
\usepackage{multicol}
\usepackage{multirow}
\usepackage{subcaption}
\usepackage{thmtools}
\usepackage{braket}
\newcommand{\ketbra}[2]{|#1\rangle\langle#2|}
\usepackage{thm-restate}

\allowdisplaybreaks
\newcolumntype{C}{>{\centering\arraybackslash}p{0.27\textwidth}}

\newcommand{\dims}{d}

\newcommand{\img}{\text{i}}

\newcommand{\povmset}{{\mathfrak{M}}}
\newcommand{\nqubits}{{N}}

\newcommand{\pauliX}{{X}}
\newcommand{\pauliY}{{Y}}
\newcommand{\pauliZ}{{Z}}
\newcommand{\pauliObsSet}{{\mathcal{P}}}

\newcommand{\opnorm}[1]{{\left\|#1\right\|}_{\text{op}}}
\newcommand{\tracenorm}[1]{{\left\|#1\right\|}_{1}}
\newcommand{\hsnorm}[1]{{\left\|#1\right\|}_{\text{HS}}}
\newcommand{\barDelta}{{\overline{\Delta}}}
\newcommand{\ptb}{{z}}
\newcommand{\ptbDistr}{{\mathcal{D}_{\ell,\cd}}}

\newcommand{\cd}{{c}}

\newcommand{\x}{\mathbf{x}}
\newcommand{\out}{{x}}

\def\multiset#1#2{\ensuremath{\left(\kern-.3em\left(\genfrac{}{}{0pt}{}{#1}{#2}\right)\kern-.3em\right)}}

\newcommand{\ham}[2]{\operatorname{d}_{\rm Ham}(#1,#2)}

\newcommand{\qmm}{{\rho_{\text{mm}}}}

\newcommand{\zest}{{\hat{z}}}

\newcommand{\HH}{\mathbb{H}}
\newcommand{\Herm}[1]{{\HH_{#1}}}

\newcommand{\qbit}[1]{|#1\rangle}
\newcommand{\qadjoint}[1]{\langle#1|}
\newcommand{\qproj}[1]{\qbit{#1}\qadjoint{#1}}
\newcommand{\qoutprod}[2]{\qbit{#1}\qadjoint{#2}}

\newcommand{\hdotprod}[2]{\left\langle#1,#2\right\rangle}
\newcommand{\matdotprod}[3]{\langle#1|#2|#3\rangle}

\newcommand{\eye}{\mathbb{I}}

\newcommand{\VecOp}{\text{vec}}
\newcommand{\vvec}[1]{|#1\rangle\rangle}
\newcommand{\vadj}[1]{\langle\langle#1|}
\newcommand{\vvdotprod}[2]{\langle\langle#1\mid#2\rangle\rangle}

\newcommand{\Luders}{\mathcal{H}}

\newcommand{\Choi}{{\mathcal{C}}}

\newcommand{\hbasis}{{\mathcal{V}}}
\newcommand{\qest}{{\hat{\rho}}}

\newcommand{\POVM}{\mathcal{M}}

\title{Pauli Measurements Are Near-Optimal for Single-Qubit Tomography}

\author{
    \begin{tabular}[t]{C@{\extracolsep{6.5em}} C}
   Jayadev Acharya &Abhilash Dharmavarapu \\
 Cornell University & Cornell University\\ 
\small \texttt{acharya@cornell.edu} &\small \texttt{ad2255@cornell.edu} 
\end{tabular}
\vspace{2ex}\\
\begin{tabular}[t]{C@{\extracolsep{6.5em}} C}
    Yuhan Liu & Nengkun Yu \\
Rice University & Stony Brook University\\ 
\small \texttt{yuhan-liu@rice.edu} &\small \texttt{nengkun.yu@cs.stonybrook.edu} 
\end{tabular}
}

\begin{document}

\maketitle

\begin{abstract}
 We provide the first non-trivial lower bounds for single-qubit tomography algorithms and show that at least ${\Omega}\left(\frac{10^N}{\sqrt{N} \varepsilon^2}\right)$ copies are required to learn an $N$-qubit state  $\rho\in\mathbb{C}^{d\times d},d=2^N$ to within $\eps$ trace distance. Pauli measurements, the most commonly used single-qubit measurement scheme, have recently been shown to require at most $O\left(\frac{10^N}{\varepsilon^2}\right)$ copies for this problem. Combining these results, we nearly settle the long-standing question of the complexity of single-qubit tomography.



\end{abstract}

\section{Introduction}

Quantum tomography is a fundamental problem in quantum information science~\cite[p.42]{Nielsen:2011:QCQ:1972505}. Given $\ns$ copies of an $\nqubits$-qubit state $\rho$, the objective is to perform quantum measurements and output an estimate $\hat \rho$ such that the trace distance between $\rho$ and $\hat \rho$ is at most $\eps$ with high probability. 

Entangled measurements that operate directly on $\rho^{\otimes\ns}$ are the most powerful, and it is now well established that their optimal copy complexity scales as $\Theta(4^\nqubits/\eps^2)$~\cite{HaahHJWY17,ODonnellW16}. While optimal copy complexity is achieved, entangled measurements are simply impractical even for moderate values of $\nqubits$. A more practical approach to problems in quantum statistical inference is via unentangled measurements, where quantum measurements are performed only on one copy of $\rho$ at a time. It has been established that the optimal copy complexity for single-copy quantum state tomography with unentangled measurements is $\Theta(8^\nqubits/\eps^2)$~\cite{KRT14, HaahHJWY17,chen2023does}. Although much simpler than entangled measurements, implementing single-copy measurements remains highly challenging. In comparison, single-qubit measurements, a subclass of single-copy measurements, are significantly easier to implement~\cite{Tabia2012, BianLQZZRSBX2015, BentQTSL2015, ChenBJZG2007, GarciaMRSTBPMTM2021}. Among single-qubit schemes, Pauli measurements are arguably among the most experiment-friendly. Pauli measurements involve measuring each qubit in the eigenbasis of one of the three $2\times 2$ Pauli operators $\pauliX,\pauliY,\pauliZ$.
 \cite{guctua2020fast} proved an upper bound $O(\nqubits\cdot 12^\nqubits/\eps^2)$. Later,
an improved upper bound of $O(10^\nqubits/\eps^2)$ was established in~\cite{Yu2020Pauli,de_Gois_2024,ADLY2025Paulinot} by observing that high-weight Pauli matrices play a more crucial role in state tomography. A lower bound of $\Omega(9.118^\nqubits/\eps^2)$ was established for Pauli measurements in~\cite{ADLY2025Paulinot}. Despite extensive study, the following question remains open.
\begin{center}
Is $\mathbf{10^\nqubits}$ the copy complexity of quantum state tomography using Pauli measurements?
\end{center}
If so, this differs significantly from entangled and single-copy measurements, where complexities scale as \( d^k \) for some integer \( k=2,3 \) with $d=2^\nqubits$ as $\log_2 10= 3.32192809489\ldots$.

Moreover, even if true, proving a matching lower bound must overcome two challenges:  1. Handling correlated measurement outcomes due to the correlation between Pauli measurements;
2. Constructing a hardness case. Notably, many existing hardness constructions rely on Haar-random unitaries \cite{HaahHJWY17,BubeckC020}. Applying this approach would treat all Pauli matrices equally, whereas high-weight Pauli matrices play a more crucial role in Pauli tomography.

A natural and more challenging question is whether general single-qubit measurements can provide an advantage over Pauli measurements. In particular, 

\begin{center}
What is the complexity of \textbf{general single-qubit} measurement schemes?
\end{center}

The flexibility in choosing single-qubit measurements introduces numerous parameters and potentially highly correlated distributions. To the best of our knowledge, no lower-bounds exist for single-qubit measurements beyond what follows from single-copy measurements.

Our recent work~\cite{ADLY2025Paulinot} focuses exclusively on the copy complexity of Pauli tomography. In particular, their (non-tight) lower bounds are shown specifically for Pauli measurements. We not only generalize their work to consider all single-qubit measurements, but also establish much stronger bounds.

\subsection{Our results}

We resolve both the questions above by establishing a lower bound for any single-qubit measurement scheme that matches the upper bound for Pauli measurement schemes up to a poly-logarithmic factor ($\sqrt{\log(2^\nqubits)}=\sqrt{N}$). Moreover, the lower bound also holds for adaptively chosen single-qubit measurements, thereby establishing that adaptivity cannot significantly improve the performance. 

\begin{theorem}
\label{thm:main-result}
 Any adaptively chosen single-qubit measurement scheme for quantum tomography that estimates an $\nqubits$-qubit state $\rho$ up to trace distance of $\eps$ with probability at least 0.9 requires at least
    \[
    \ns=\bigOmega{\frac{10^\nqubits}{\sqrt{\nqubits}\eps^2}}
    \]
    copies of $\rho$.
\end{theorem}

Combining with the upper bounds for Pauli measurements, the following corollary is immediate.
\begin{corollary}
Pauli measurements are near-optimal for single-qubit quantum tomography.
\end{corollary}

\subsection{Our techniques}
We give an overview of our main novelties and elaborate on them in later sections.
\paragraph{Less is more.} Lower bound proofs usually involve constructing a set of instances (which are quantum states in our problem) and a decision problem that is hard for all learning algorithms. In nearly all prior works on statistical parametric estimation, both quantum and classical, the hard instances are constructed from a subspace with roughly the same dimension as the ambient dimensionality of the parameter space: $\Theta(\dims^2)= 4^\nqubits$ for full-rank states~\cite{HaahHJWY17,ODonnellW17,BubeckC020, chen2023does, ADLY2025Paulinot}, $\Theta(\dims)$ for $\dims$-dimensional classical distributions such as discrete distribution, \cite{Paninski08, duchi2013local, ACLST22iiuic}, Gaussian \cite{acharya2020distributed} and product Bernoulli distributions~\cite{Canonne2020testing}. It seems like no other choice is reasonable, as we want to take full advantage of the complexity of the problem in our hard case.

The most surprising insight that leads to the near-optimal lower bounds almost seems sacrilegious to this belief. For some $\beta<1$,
\begin{center}
    \fbox{We only need $O(4^{\nqubits\beta})$ degrees of freedom in the hard construction for single-qubit measurements.}
\end{center}
 This is strictly sublinear in the dimensionality of the ambient space, which is $\Theta(4^\nqubits)$. We will elaborate on how we came to this conclusion in \cref{sec:why10}.

\paragraph{Handling arbitrary single-qubit measurements.} Although being an important step, the argument in~\cite{ADLY2025Paulinot} was specifically tailored to Pauli measurements. 
The lower bound construction perturbs the maximally mixed state along different Pauli observables, which happen to be the eigenvectors of the post-measurement channel/measurement information channel~\cite{liu2024restricted} of Pauli measurements. This seems like a happy coincidence for Pauli measurements.
 
 Moreover, the cardinality of Pauli measurements is only $3^{\nqubits}$, much smaller than the dimension of full-rank states $4^\nqubits$. It is not clear whether their method can extend to the set of single-qubit measurements, which is uncountably infinite. 

 We generalize the techniques in \cite{ADLY2025Paulinot} to single-qubit measurements by establishing upper bounds on the the mutual information of the \emph{measurement information channels} for any single-qubit measurement scheme.
 It turns out that the mutual information upper bound can be achieved with Pauli measurements. 
 An intuitive description is in \cref{sec-mic-single-qubit-overview} and a detailed derivation is in \cref{sec:MI-single-qubit}.

 \subsection{Related works}

In a concurrent work, Keenan, Goold, and Nico-Katz used Gaussian unitary ensembles to design a non-adaptive algorithm based on Pauli measurements, recovering the $10^{\nqubits}/\varepsilon^2$ upper bound in \cite{keenan2025randommatrixtheorypauli}. Their results suggest that this bound is the best achievable for their method in the absence of a rigorous proof for general non-adaptive measurements.

Quantum tomography has also been studied under low rank assumptions. In particular,~\cite{HaahHJWY17, ODonnellW16} showed that the optimal copy complexity of estimating a $d$-dimensional (note $d=2^N$) quantum state with rank $r$ is $\Theta(dr/\eps^2)$ with entangled measurements. With non-adaptive single-copy measurements, the copy complexity is known to be $\Theta(dr^2/\eps^2)$. Interpolating between fully-entangled and single-copy measurements,~\cite{Chen0L24memory} established tight bounds as a function of the number of copies that we are allowed to make an entangled measurement on. 


Quantum tomography has also been studied under other distance measures such as fidelity~\cite{HaahHJWY17, Chen0HL22, Yuen_2023}, and quantum relative entropy~\cite{flamian2023tomography}.

Pauli observables, a special class of Pauli measurements, have also been studied.~\cite{GrossLFBE10, Flammia_2012} established the sample complexity of rank-$r$ tomography under non-adaptive measurements, and~\cite{cai2016optimal} obtained near-optimal error rates under Hilbert Schmidt and operator norm induced distance. 

Quantum state certification is the related problem of hypothesis testing where the goal is to test whether $\rho$ is equal to a reference state. The problem has been studied with entangled measurements in~\cite{ODonnellW15, BadescuO019}, with single copy measurements in~\cite{BubeckC020, Chen0HL22, LiuA24}, and with Pauli measurements~\cite{Yu2023almost}. 

Quantum shadow tomography and quantum overlapping tomography study the situation where we are interested in obtaining some partial information about the quantum state. \cite{Cotler_2020, Garc_a_P_rez_2020, evans2019scalable} consider the problem of obtaining simultaneously all the $k$-qubit reductions of an $\nqubits$-qubit state. \cite{Aaronson20, Huang_2020, ChenCH021,yu2023learningmarginalssuffices} study shadow tomography to estimate expectations of observables and~\cite{chen2024pauli} study Pauli shadow tomography with constrained measurements. 

The setting of single-copy quantum inference is similar to distributed statistical inference, where the goal is to solve a statistical task when the samples are distributed across users who can only send an information-constrained (e.g., communication, or privacy constraints) message about their sample. These problems have been extensively studied in the past decade, and we draw from some of the lower bound methods developed in these works and generalize them to the quantum framework~\cite{duchi2013local, barnes2019lower, AcharyaCT19, ACT:19:IT2, AcharyaCST23}.

\paragraph{Organization.} The rest of the paper is organized as follows. In \cref{sec:why10} we motivate the appearance of 10 in the copy complexity of Pauli tomography. In \cref{sec:overview} we provide an overview of our lower bound. \cref{sec:preliminaries} provides the formal problem setup, describes quantum measurements, and discusses various probability divergences. \cref{sec:hard-case} describes the construction of the hard case for our lower bound. In \cref{sec:lower-final} we prove the lower bound.

\section{Why ten?}\label{sec:why10}
It may seem surprising that the number $10$ appears in the sample complexity of Pauli tomography. A qubit resides in a two-dimensional Hilbert space, with three Pauli operators per qubit. Thus, one might expect the tight sample complexity to be related to powers of two and/or three. In this context, ten is peculiar because it includes 5 as a prime factor—a seemingly unrelated number.

In this paper, we argue that $10$ is very natural for Pauli tomography. 
It's just that we should not look at the prime factors. For our problem, the right way to decompose the number is \[\mathbf{10=1+3^2}.\]  Thus, when raising it to the power of $N$, using the binomial theorem,
\begin{equation}
\label{equ:10-binom}
    10^{\nqubits} =(3^2+1)^\nqubits=\sum_{m=0}^\nqubits\binom{\nqubits}{m}(3^2)^{m}.
\end{equation}
In other words, $10^N$ is a combination of even powers of 3 and binomial coefficients, which now appears more naturally related to our problem: $\mathrm{3}$ comes from Pauli measurements. 
This simple expression  (or at least a partial sum of it) is the magic equation that arises in the proofs of both the upper and lower bounds.

\subsection{Review of the algorithm: $1+{3}^{\mathrm{2}}$ comes from $\ell_2$ norm analysis}
\label{sec:alg-review}
We now review the algorithm that achieves the $10^{\nqubits}$ sample complexity in \cite{Yu2020Pauli,de_Gois_2024, ADLY2025Paulinot}. Through the analysis, we not only observe how 10 arises in Pauli tomography, but also gain an important insight on how to close the gap in the lower bound.

Since there are a total of $3^\nqubits$ Pauli measurements, we evenly divide $\ns$ copies into $3^\nqubits$ groups of equal size $m=\ns/3^\nqubits$, each group measures the copies with the Pauli measurement associated with $P=\sigma_{1}\otimes\cdots \otimes \sigma_{\nqubits}$ where $\sigma_i\in \{X, Y, Z\}$. From the binary strings obtained from the measurements, we can obtain empirical estimates for the expectation value of each Pauli observable $Q\in \{X, Y, Z, \eye_2\}^{\otimes \nqubits}$, $\Tr[\rho Q]$.

Let $w(Q)$ be the weight of $Q$, which is the number of non-identity components in $Q$. We can learn about $\Tr[\rho Q]$ for all Pauli measurement $P$ that match the non-identity components in $Q$. For example, if $Q=X\otimes X\otimes \eye_2^{\otimes(\nqubits-2)}$, then all $P$ with $\sigma_1=\sigma_2=X$ can be used to learn about $\Tr[\rho Q]$ by multiplying the $\{-1, 1\}$ outcome from the first 2 qubits. For $Q$ with weight $w$, the number of such useful Pauli measurements is
\[
n(Q)= 3^{\nqubits-w(Q)}.
\]
Thus we get $n(Q)m$ i.i.d. $\{-1,1\}$ binary samples to estimate $\Tr[\rho Q]$. The empirical average $E(Q)$ has a variance of at most $1/(n(Q)m)$,
\begin{equation}
\label{equ:error-single-pauli}
    \expect{(E(Q)-\Tr[\rho Q])^2}\le \frac{1}{n(Q)m}=\frac{3^{w(Q)}}{\ns}.
\end{equation}
Note that a quantum state can be seen as a linear combination of Pauli observables (including $\eye_\dims$),
\[
\rho = \sum_{Q}\frac{\Tr[\rho Q]}{\dims}Q.
\]
Our estimate $\hat{\rho}$ is then $\hat{\rho}=\sum_{Q}\frac{E(Q)}{\dims}Q$.
Since Pauli observables are orthogonal, we can compute the Hilbert-Schmidt distance in terms of the $\ell_2$ distance of expectation values,
\begin{align*}
    \expect{\hsnorm{\rho-\hat{\rho}}^2}&=\frac{1}{\dims}\sum_{Q}\expect{(E(Q)-\Tr[\rho Q])^2}\\
    &\leq\frac{1}{\dims}\sum_{Q}\frac{3^{w(Q)}}{\ns}\\
    &=\frac{1}{\dims}\sum_{w=0}^{\nqubits}\binom{\nqubits}{w}3^w \cdot \frac{3^w}{\ns}\\
    &=\frac{1}{\dims \cdot \ns}\sum_{w=0}^{\nqubits}\binom{N}{w}9^w\\
    &=\frac{10^\nqubits}{\dims\ns}.
\end{align*}
The key is the second step, where we split the summation according to the weight of Pauli observables. Then the combinatorial Equation \eqref{equ:10-binom} naturally pops up. Using Cauchy-Schwarz and Jensen's inequality,
\begin{align*}
    \Paren{\expect{\tracenorm{\rho-\hat{\rho}}}}^2\le \expect{\tracenorm{\rho-\hat{\rho}}^2}\le \dims \expect{\hsnorm{\rho-\hat{\rho}}^2}\le \frac{10^\nqubits}{\ns}.
\end{align*}
Setting the right-hand-side to be $\eps^2$ gives the desired bound for $\ns$.

From the above analysis, especially \eqref{equ:error-single-pauli}, we can see that Pauli observables with larger weights contribute more to the final error, which grows exponentially with weight.

\subsection{Less is more: weakness in prior works}
\label{sec:10-lower}
\cite{ADLY2025Paulinot} exploits the weakness of Pauli measurements and obtains the first non-trivial lower bound specifically for Pauli measurements. 
 In that lower bound, we explicitly constructed perturbations along roughly $\ell=\dims^2/2$ different Pauli directions with the largest weights,
 \begin{equation}\label{equ:perturbation-informal}
     \sigma_z=\qmm+\frac{\eps}{\sqrt{\dims\ell}}\sum_{i=1}^\ell z_i V_i, \quad z\sim\{-1,1\}^\ell,
 \end{equation}
where $V_i$'s are normalized Pauli observables sorted in decreasing weight. 

It seems natural to choose $\ell\sim \dims^2$ since density matrices have dimension $\dims^2-1$, and to make the hard case as difficult as possible, it may be best to exploit all dimensions of freedom available. In fact, to our knowledge, nearly all previous works on tomography and state testing use constructions with roughly $\dims^2$ degrees of freedom for full-rank states. This includes Haar random unitary rotations \cite{HaahHJWY17,ODonnellW16,BubeckC020}, and Gaussian orthogonal ensembles \cite{Chen0HL22, chen2023does}. 

However, this is not necessarily true for restricted settings such as single-qubit measurements. The algorithm in~\cref{sec:alg-review} provides a strong hint on why that's the case: choosing $\ell =o(\dims^2)$ suffices to show that learning is hard using Pauli measurements.

To see this, let's ask a question: does there exist a state that lies in an $o(\dims^2)$-dimensional subspace, but the algorithm still incurs roughly the same error? Let us recall the expected Hilbert-Schmidt error and write things slightly differently,

\[
\expect{\hsnorm{\rho-\hat{\rho}}^2}=\frac{1}{\dims \ns}\sum_{w=0}^{\nqubits}\binom{N}{w}9^w=\frac{10^\nqubits}{\dims \ns}\sum_{w=0}^{\nqubits}\binom{N}{w}\Paren{\frac{9}{10}}^w\Paren{\frac{1}{10}}^{\nqubits-w}.
\]

The summation trivially equals 1. However, what if the summation instead starts at $9\nqubits/10$? Then,
\[
\sum_{w=9\nqubits/10}^{\nqubits}\binom{N}{w}\Paren{\frac{9}{10}}^w\Paren{\frac{1}{10}}^{\nqubits-w}=\probaOf{\binomial{\nqubits}{\frac{9}{10}}\ge \frac{9\nqubits}{10}}\simeq \frac{1}{2},
\]
which is still at least a constant (since the median of the binomial distribution is roughly the same as its mean). Therefore, even if a state $\rho$ only contains Pauli observables with weight at least $9\nqubits/10$, our analysis cannot do much better than $10^{\nqubits}/\eps^2$! The degree of freedom for this set of quantum states (i.e., the number of Pauli observables with weight at least $0.9\nqubits$) would be
\[
\sum_{w=9\nqubits/10}^{\nqubits}\binom{\nqubits}{w}3^w = \bigO{\dims^{1.9}}.
\]

This strongly suggests that we should only choose Pauli observables with $w\ge 0.9\nqubits$ in the construction since they already contribute to nearly half of the error in the algorithm, and thus $\ell=o(\dims^2)$. We will later show that this observation is not due to a weak upper bound analysis, but a fundamental constraint present in all \emph{single-qubit} measurements.

\section{Why is it technically nontrivial?}
\label{sec:overview}
As suggested by \cref{sec:10-lower}, we might need to choose $\ell$ much smaller than $\dims^2$. However, one important technical issue arises: would \cref{equ:perturbation-informal} even be a valid quantum state with high probability? 

Another issue is to argue that the same lower bound holds for arbitrary single-qubit measurements. The lower bound construction is specifically tailored to Pauli measurements and may not be suitable for other single-qubit measurements.

In this section, we provide an overview of how to resolve these issues.

\subsection{Better matrix concentration--only work for $\eps=O(1/\log(d))$}
To guarantee that the construction is valid, existing works \cite{ADLY2025Paulinot,LiuA24,liu2024restricted} choose $\ell\sim \dims^2$.
Fortunately, this is not always necessary, at least for some parameter regimes.

For \eqref{equ:perturbation-informal} to be a valid state, it suffices to argue that $\opnorm{\sigma_z-\qmm}=O(1/\dims)$ with high probability. Using matrix Bernstein/Chernoff bounds, we can achieve the desired bound up to a $\log(\dims)$ factor.

\begin{theorem}[Matrix Bernstein inequality \cite{Tropp_2011}]
    Let $X_1, \ldots, X_{n}\in \C^{\dims\times\dims}$ be independent random Hermitian matrices with $\expect{X_i}=0$, $\opnorm{X_i}\le B$, and $\sigma^2\eqdef \opnorm{\expect{\sum_{i}X_i^2}} $. Then for all $t\ge 0$,
    \[
    \probaOf{\opnorm{\sum_{i=1}^n X_i}\ge t}\le \dims \exp\Paren{-\frac{t^2}{\sigma^2 + Bt/3}}.
    \]
\end{theorem}

For any normalized Pauli $V_i$, $B=\opnorm{V_i}=1/\sqrt{\dims}$. Furthermore, since $V_i^2=\eye_\dims/\dims$, we have for $\ell$ normalized Paulis, $\sigma^2=\ell/\dims$. Therefore,

\[
\probaOf{\opnorm{\sum_{i=1}^\ell z_iV_i}\ge t}\le \dims \exp \Paren{-\frac{t^2}{\ell/\dims + t/(3\sqrt{\dims})}}.
\]
Choosing $t=10\sqrt{\frac{\ell}{\dims}}\log(\dims)$, and using $\dims<\ell <\dims^2$, we have with probability at least $1-1/\poly(\dims)$,
\[
\opnorm{\sigma_z-\qmm}\le \bigO{\frac{\eps}{\dims}\log(\dims)}.
\]
Thus, when $\eps=O(1/\log(\dims))$, we have a valid state with high probability. We can further show using H\"older inequality that with the same probability, $\tracenorm{\sigma_z-\qmm}\ge\eps$. Thus, we have a valid construction that is $\eps$-far from the target state $\qmm$.

Is the additional logarithmic requirement on $\eps$ necessary?

We will employ a more advanced matrix concentration result that better takes advantage of the non-commutativity of Pauli matrices. Furthermore, the failure probability can be decreased to be super-polynomially small in $\dims$. 
See \cref{sec:hard-case}.

\subsection{Measurement information channel for single-qubit measurements}
\label{sec-mic-single-qubit-overview}
We provide an intuition for why Pauli measurements should be the best among single-qubit measurements. 
The Pauli bases, $\{\qbit{0}, \qbit{1}\}, \{\qbit{+}, \qbit{-}\}, \{\qbit{+\img}, \qbit{-\img}\}$ form a 2-design in the space of qubits. 
Thus, when we evenly assign all Pauli measurements, it is almost like performing a Haar-random basis measurement on each qubit. However, to rigorously prove the optimality of Pauli measurements requires some work.

The measurement information channel (MIC)\cite{liu2024restricted} is a powerful tool to quantify the distinguishability of measurements.
\begin{definition}
    Let $\POVM=\{M_x\}_x\in \cX$ be a measurement. The measurement information channel is defined as
\begin{equation}\label{equ:mic}
    \Luders_{\POVM}(A)\eqdef\sum_{x}M_x\frac{\Tr[M_xA]}{\Tr[M_x]}, \quad\Choi_{\POVM}\eqdef \sum_{x}\frac{\vvec{M_x}\vadj{M_x}}{\Tr[M_x]} \in \C^{\dims^2\times\dims^2},
\end{equation}
where $\vvec{M_x}=\VecOp(M_x)$ and $\vadj{M_x}=\VecOp(M_x)^{\dagger}$. $\Luders_{\POVM}$ is its Kraus representation whereas $\Choi_{\POVM}$ is the matrix representation that satisfies $\vvec{\Luders_{\POVM}(\rho)}=\Choi_{\POVM}\vvec{\rho}$.
\end{definition}

When $\POVM$ is a projective measurement, MIC has a physical meaning: it is exactly the quantum channel/instrument that describes the post-measurement state transition when the outcome is lost.

As a toy example, consider single qubit states $\rho=\eye_2/2$ and $\phi_z = \eye_2/2 + \alpha(z_X X+z_Y Y+z_Z Z)$, where $z_X, z_Y, z_Z$ are drawn uniformly from $\{-1,1\}^3$, $\alpha$ is chosen sufficiently small to make $\phi$ a valid quantum state. We are essentially mimicking the construction \eqref{equ:perturbation-informal} for single-qubit states.

For simplicity, consider a single-qubit basis measurement $\POVM=\{\qproj{u_1}, \qproj{u_2}\}$. Using MIC, we can compute the distance of post-measurement states of $\rho$ and $\phi_z$. After measuring $\rho$ with $\POVM$, the post-measurement state is still $\eye_\dims/2$. For $\phi$, the post-measurement state is
\[
 \Luders_{\POVM}(\phi)=\sum_{i=1,2}\qproj{u_i}\matdotprod{u_i}{\phi}{u_i}=\frac{\eye_2}{2}+\alpha\sum_{i=1,2}\qproj{u_i}\matdotprod{u_i}{X+Y+Z}{u_i}.
\]
The Hilbert-Schmidt distance is
\[
\hsnorm{\Luders_{\POVM}(\rho-\phi_z)}^2=\alpha^2\sum_{i=1,2}\matdotprod{u_i}{z_X X+z_Y Y+z_Z Z}{u_i}^2,
\]
which is also the $\ell_2$ distance and chi-square divergence (up to constants) between outcome distributions. Since  $\dims =2$, the total-variantion distance is within a constant factor of $\ell_2$ distance, so the above quantity characterizes distinguishability between two states. Note that $z_X, z_Y, z_Z$ are chosen independently from $\{-1, 1\}$. In expectation,
\[
\expectDistrOf{z}{\hsnorm{\Luders_{\POVM}(\rho-\phi_z)}^2}=\alpha^2\sum_{i=1,2}(\matdotprod{u_i}{X}{u_i}^2+\matdotprod{u_i}{Y}{u_i}^2+\matdotprod{u_i}{Z}{u_i}^2)=2\alpha^2.
\]
The final step is because $\eye_2, \sigma_X, \sigma_Y, \sigma_Z$ form an orthogonal basis for $2\times 2$ complex matrices and $\Tr[(\qproj{u_i})^2]=1$, so $\matdotprod{u_i}{X}{u_i}^2+\matdotprod{u_i}{Y}{u_i}^2+\matdotprod{u_i}{Z}{u_i}^2=1$. 

Therefore, regardless of the choice of the basis measurement, the expected distance between the post-measurement states of $\rho$ and $\phi_z$ is still close. This serves as a simple but intuitive example of how MIC characterizes the power of any single-qubit measurement.

In \cref{sec:MI-single-qubit}, we provide a complete argument for arbitrary single-qubit POVMs and how our construction using Pauli observables still suffices to show a tight lower bound for single-qubit measurements.

\section{Preliminaries}
\label{sec:preliminaries}
\subsection{Quantum States and Measurement}
\paragraph{Quantum states.}
A \textit{pure state} is presented by a unit vector $\ket{\psi} \in \C^\dims$. $\ket{\psi}$ is in the span of the orthonormal basis of the Hilbert space,
\begin{align*}
    \ket{\psi} = \sum_{i=1}^\dims \alpha_i, \ket{e_i} \quad \sum_{i=1}^\dims |\alpha_i|^2 = 1.
\end{align*}
Measuring the state along the basis results in observing $\ket{e_i}$ state with probability $|\alpha_i|^2$. \\
More generally, one can write a statistical mixture of pure states, known as mixed states. 
\begin{align*}
    \rho = \sum_{i=1}^\dims p_i \ket{\psi_i} \bra{\psi_i},
\end{align*}
where the pure state $\ket{\psi_i}$ is observed with probability $p_i$. Mixed states have the following properties:
\begin{align*}
    \rho \succeq 0, \quad \Tr[\rho] = 1.
\end{align*}

\paragraph{POVMs.}
Measurements are typically represented using \emph{Positive Operator-Valued Measure} (POVM). They are represented by a collection of p.s.d observables that add up to the identity.
\begin{align*}
    \POVM = \left\{M_x\right\}_{x \in \mathcal{X}}, \quad \sum_{x \in X} M_x = \eye_d.
\end{align*}
Born's rule states that the probability of observing outcome $x$ for state $\rho$ is
\begin{align*}
    \Pr[X=x] = \Tr[M_x \rho].
\end{align*}
Moreso, Born's rule states that $\rho$ collapses to some state after measurement. In some cases, the measurement is repeatable on the same state if $M_x$ are all orthogonal. However, this is not typically the case. Thus, the copy-wise tomography task considers \emph{measuring each copy once}.
\subsection{Problem setup}
Given $n$ copies of $\rho$, the goal is to design a measurement $\POVM^n = \{\POVM^i\}_{i=1}^n$ and an estimator $\hat{\rho}:\cX^{\ns}\mapsto\C^{\dims\times\dims}$ such that
\begin{align*}
    \inf_{\rho} \Pr[\tracenorm{\hat{\rho}(x^n) - \rho} \leq \eps] \geq \frac{99}{100},
\end{align*}
where $\tracenorm{\cdot}$ is the trace distance described in \cref{sec:hilbert}, and $x^\ns=(x_1, \ldots, x_\ns)$ are the measurement outcomes. Furthermore, we allow adaptivity in measurement, so the measurement on one copy can depend on the outcome of previously measured copies, 
\begin{align*}
   \p_\rho^{x_i|x^{i-1}}(x) = \Tr[M^i_x \rho].
\end{align*}
\subsection{Hilbert space of linear operators}
\paragraph{Hilbert space.} The Hilbert space over complex matrices is the set of all $A \in \C^{\dims \times \dims}$ equipped with the Hilbert-Schmidt inner product $\hdotprod{A}{B} = \Tr[A^\dagger B]$. The set of Hermitian matrices is in the span of complex matrices with real coefficients; thus, Hermitian matrices are described by a real Hilbert space with the same inner product.
\paragraph{Super-operators.} A linear super-operator $\mathcal{N}(\cdot)$ describes the linear mapping between operators. It is often convenient to describe the super operator in respect to the linear isomorphism $\VecOp(\qoutprod{i}{j})\eqdef \qbit{j}\otimes \qbit{i}$, where $\{\ket{j}\}_{j=1}^d$ is the computational basis. We will call $\vvec{A} = \VecOp(A)$ for convenience. In such case, $\vvec{\mathcal{N}(X)} = C(\mathcal{N})\vvec{X}$. Where $C(\mathcal{N})$ is the matrix representation of $\mathcal{N}$ in $\C^{\dims^2 \times \dims^2}$.
\paragraph{Schatten norms.} \label{sec:hilbert}
Consider a linear operator $A \in \C^{\dims \times \dims}$ with singular values $(\lambda_1,..., \lambda_d)$. The \emph{p-schatten norm} ($\|\cdot\|_{S_p}$) of $A$ is the $\ell_p$ norm of the vector $(\lambda_1,..., \lambda_d)$.
Furthermore, for hermitian $B$, the p-schatten norm represents the $\ell_p$ norm of the absolute value of the eigenvalues.
Examples of the Schatten norms include
\begin{equation}
    \hsnorm{A} = \sqrt{\Tr[A^\dagger A]} = \|A\|_{S_2},
    \; \; \tracenorm{A} = \Tr[|A|] = \|A\|_{S_1},
    \; \; \opnorm{A} = \max_{1 \leq i \leq d} |\lambda_i| = \|A\|_{S_\infty}.
\end{equation}
\subsection{Pauli observables and basis measurements}
Pauli basis measurements measure along the direction of Pauli observables, which are represented by the tensor product of Pauli matrices (along with $\eye$),
\begin{equation}
    X = \begin{bmatrix}
    0 & 1 \\
    1 & 0
\end{bmatrix},\quad
    Y = \begin{bmatrix}
    0 & -i \\
    i & 0
\end{bmatrix},\quad
    Z = \begin{bmatrix}
    1 & 0\\
    0 & -1
\end{bmatrix}.
\label{equ:pauli-ops}
\end{equation}
Let $\Sigma = \{\frac{1}{\sqrt{2}} X, \frac{1}{\sqrt{2}} Y. \frac{1}{\sqrt{2}} Z\}$ be the normalized Pauli Matrices. Due to $\Sigma \cup \{\frac{\eye}{\sqrt{2}}\}$ being an orthonormal basis for quantum states, $\mathcal{V} = \{\Sigma \cup \{\frac{\eye}{\sqrt{2}}\}\}^{\otimes N}$ form an orthonormal basis for $2^N$-dimensional mixed states,
\begin{equation}
    \rho = \frac{\eye_\dims}{\dims} +\sum_{P\in\pauliObsSet}\alpha_PP, \quad\alpha_P=\frac{\Tr[\rho P]}{\dims},
    \label{equ:pauli-decomposition}
\end{equation}
where $\pauliObsSet = \mathcal{V} \setminus \frac{\eye_\dims}{\dims}$. 
Thus, measuring along an observable $P$ provides information about the coordinate $\alpha_P$. $\POVM^P = \{M_{-1}, M_1\}=\{\frac{\eye_d-P}{2}, \frac{\eye_d+P}{2}\}=\left\{\qproj{u^+_{P}}, \qproj{u^-_{P}}\right\}$ corresponds to a \textit{Pauli observable Measurement}. \\ \\
A Pauli basis measurement for observable $P = \sigma_1 \otimes ... \otimes \sigma_N$ is indexed by $x \in \{-1,1\}^N$ and considers the tensor product of the observable measurements,
\begin{equation}
\label{def:Pauli_basis}
    M_x = \bigotimes_{i=1}^N \frac{\eye + x_i \sigma_i}{2} \quad,\quad
    \Tr[M_x P'] = \prod_{i=1}^N \left(x_i \indic{\sigma_i = \sigma'_i} +  \indic{\sigma_i = \eye}\right).
\end{equation}
The difficulty of learning about the coordinate $\alpha_P$ in \cref{equ:pauli-decomposition} is characterized by the weight of P,
\begin{equation}
    w(P) = \sum_{i=1}^N \indic{\sigma_i \neq \eye}.
\end{equation}
For lower weight Pauli observables, there exists multiple measurements where $\Tr[M_x P']$ are non-zero. Precisely, there are $3^{N-w}$ measurements that give us information about $\alpha_{P^{'}}$. Notice that weight $N$ observables require that the corresponding observables between the measurement and $P'$ match for every qubit, so there exists only one measurement that can learn $\alpha_{P^{'}}$.
\subsection{Probability divergences}
We will be bounding changes in post-measurement distributions when a quantum state is perturbed. Thus, it will be helpful to discuss relevant probability divergences. Let $\p$ and $\q$ be discrete distributions over $\mathcal{X}$. The \emph{total variation distance} is defined as 
\[
\totalvardist{\p}{\q}\eqdef\sup_{S\subseteq\mathcal{X}}(\p(S)-\q(S))=\frac{1}{2}\sum_{x\in\mathcal{X}}|\p(x)-\q(x)|.
\]
The \emph{KL-divergence}  is (for $\q$ absolutely continuous with respect to $\p$)
\[
\kldiv{\p}{\q}\eqdef\sum_{x\in\mathcal{X}}\p(x)\log\frac{\p(x)}{\q(x)}.
\]
The \emph{symmetric KL-divergence} is 
$\kldivsym{\p}{\q}:=\frac{1}{2}(\kldiv{\p}{\q}+\kldiv{\q}{\p})$. \\
The \emph{chi-square divergence} is
\[
\chisquare{\p}{\q}\eqdef \sum_{x\in\mathcal{X}}\frac{(\p(x)-\q(x))^2}{\q(x)}.
\]
By Pinsker's inequality and the concavity of the logarithm,
\[
2\totalvardist{\p}{\q}^2\le \kldiv{\p}{\q}\le \chisquare{\p}{\q}.
\]
We define $\ell_p$ distance as $
\norm{\p-\q}_p\eqdef\Paren{\sum_{x\in\mathcal{X}}{|\p(x)-\q(x)|^p}}^{1/p}.
$

\emph{Mutual information} is an essential tool in our analysis as it helps define \emph{a correlation} between the outcome distribution and the parameters in the lower-bound construction. Let $(X,Y)$ be a joint discrete random variable in $\mathcal{X}\times\mathcal{X}$ with marginals $p_X, p_Y$ and joint distribution $p_{XY}$. The mutual information is defined as
\begin{align*}
    I(X;Y) = \sum_{x,y \in \mathcal{X} \times \mathcal{X}} p_{XY}(x,y) \log \frac{p_{XY}(x,y)}{p_X(x)\cdot p_Y(y)} = \kldiv{p_{XY}}{p_X \otimes p_Y}.
\end{align*}

\section{New design of the hardness cases--the missing piece}\label{sec:hard-case}
\begin{definition}
 \label{def:perturbation}
     Let $\ell\le\dims^2-1$ and $\hbasis=(V_1, \ldots, V_{\dims^2}=\frac{\eye_\dims}{\sqrt{\dims}})$ be an orthonormal basis of $\Herm{\dims}$ (the space of $\dims\times\dims$ Hermitian matrices), and $\cd$ be a universal constant. Let  $\ptb=(\ptb_1, \ldots, \ptb_\ell)$ be uniformly drawn from $\{-1, 1\}^\ell$, we define $\sigma_{\ptb}=\qmm + \barDelta_{\ptb}$ where
     \begin{equation}
         \Delta_{\ptb} = \frac{\cd\eps}{\sqrt{\dims}}\cdot\frac{1}{\sqrt{\ell}}\sum_{i=1}^\ell \ptb_iV_i, \quad \barDelta_{\ptb}= \Delta_{\ptb}\min\left\{1, \frac{1}{2\dims \opnorm{\Delta_{\ptb}}}\right\},
         \label{equ:delta_z}
     \end{equation}
     Let the distribution of $\sigma_z$ be $\ptbDistr(\hbasis)$.
 \end{definition}
 
\begin{theorem}\label{Hard}
    Let $V_1, \ldots, V_{\dims^2-1}$ be normalized Pauli matrices and $W=\sum_{i=1}^\ell z_iV_i$ with $\ell \ge \dims^{\frac{3}{2}}$. There exists a universal constant $C$ such that with probability at least $1-\exp(-\dims)$ over uniform $z\sim \{-1, +1\}^\ell$, 
    \[
    \opnorm{W}\le C \sqrt{\frac{\ell}{\dims}}.
    \]
    Thus $\tracenorm{\sigma_z-\qmm}\ge \eps$ with probability at least $1-\dims\exp(-\ell^{\frac{1}{4}})$ for appropriately chosen $c$ in~\cref{def:perturbation}, e.g. $c\le 1/200$. 
\end{theorem}

\begin{proof}
The proof employs the following theorem,

\begin{theorem}[{\cite[Theorem 1.4]{Bandeira_2023}}]
     Let \( Z_1, \dots, Z_n \) be arbitrary independent \( d \times d \) self-adjoint centered random matrices, and let 
\[
X = \sum_{i=1}^{n} Z_i.
\]
Then,
\[
\probaOf{\opnorm{X} \geq \opnorm{X_{\text{free}}} + C \left\{ v^{\frac{1}{2}} \sigma^{\frac{1}{2}} (\log d)^{\frac{3}{4}} + \sigma^* t^{\frac{1}{2}} + R^{\frac{1}{3}} \sigma^{\frac{2}{3}} t^{\frac{2}{3}} + R t \right\}} \leq d e^{-t}
\]
for all \( t \geq 0 \), where \( C \) is a universal constant, 
\[
\sigma = \opnorm{\mathbb{E} X^2}^{\frac{1}{2}}, \quad v = \opnorm{\text{Cov}(X)}^{\frac{1}{2}}, \quad \sigma^* = \sup_{\|v\| = \|w\| = 1} \mathbb{E} [ | \langle v, X w \rangle |^2]^{\frac{1}{2}} \leq v, \quad R = \max_i \opnorm{Z_i},
\]
and \( X_{\text{free}} \) is the free model associated with the centered Gaussian random matrix whose entries have the same covariance as those of \( X \) (in particular, \( \|X_{\text{free}}\| \leq 2\sigma \)). \\ \\
Here, for any centered $d \times d$ random matrix $X$, the covariance matrix $Cov(X)$ is defined as
\[
Cov(X)_{ij,kl} = \mathbb{E}[X_{ij}X_{kl}^*].
\]
\end{theorem}
\noindent In our problem, we set $Z_i=z_i V_i$. Thus,
\begin{align*}
    &\sigma
    =\opnorm{\mathbb{E} X^2}^{\frac{1}{2}}=\opnorm{\mathbb{E} \sum_{i,j} z_iz_j V_iV_j}^{\frac{1}{2}}=\opnorm{\frac{\ell}{d} I}^{\frac{1}{2}}=\sqrt{\frac{\ell}{d}},\\
   & \|X_{\text{free}}\| \leq 2\sigma=2\sqrt{\frac{\ell}{d}}.
\end{align*}
    We can further compute $Cov(X)$ as follows,
 \begin{align*}
 Cov(X)&=\sum_{i,j,k,l}Cov(X)_{ij,kl}\otimes \ketbra{ij}{kl} \\
 &= \expectDistrOf{}{X_{ij}X_{kl}^*\otimes \ketbra{ij}{kl}}\\
 &=\expectDistrOf{}{\vvec{X}\vadj{X}}\\
 &=\expectDistrOf{}{\sum_{i,j} z_iz_j\vvec{V_i}\vadj{V_j}}\\
 &=\sum_{i} \vvec{V_i}\vadj{V_i}.
 \end{align*}
 This implies that
  \begin{align*}
    v = \opnorm{\text{Cov}(X)}^{\frac{1}{2}}=1, \  \mathrm{and}\ \      \sigma^*\leq v=1, \  \mathrm{and}\ \   R=\frac{1}{\sqrt{d}}. 
\end{align*}
As $\dims^{\frac{3}{2}}\leq \ell \leq d^2$, we will choose
\[
t= \ell^{\frac{1}{4}}.
\]
This choice guarantees that
\[
 v^{\frac{1}{2}} \sigma^{\frac{1}{2}} (\log d)^{\frac{3}{4}} + \sigma^* t^{\frac{1}{2}} + R^{\frac{1}{3}} \sigma^{\frac{2}{3}} t^{\frac{2}{3}} + R t=O(\opnorm{X_{\text{free}}}).\qedhere
\]

\end{proof}

\section{It is ten!} 
\label{sec:lower-final}
This section establishes the lower bound $\tilde{\Omega}(\frac{10^{\nqubits}}{\eps^2})$ by studying the mutual information bounds.

The key to establishing copy-complexity is to bound the average mutual information between the measurement outcome distribution and the z vector as a function of the number of copies ($n$) and the MIC. For convenience, let $\mathcal{G}$ be the set of all $z\in \{1,-1\}^\ell$ such that $\tracenorm{\sigma_z-\qmm}\ge \eps$ according the concentration in \cref{Hard}.
\begin{restatable}[{\cite[Theorem 4.4]{ADLY2025Paulinot}}]{theorem}{}
\label{thm:avg-MI-upper}
    Let $\sigma_\ptb\sim\ptbDistr(\hbasis)$ where $\ptb\sim\{-1,1\}^{\ell}$, $\out^\ns$ be the outcomes after applying $\POVM^\ns$. Then,
    \begin{align}
         \frac{1}{\ell}\sum_{i=1}^{\ell}\mutualinfo{\ptb_i}{\out^\ns}&\le \frac{8 \ns c^2 \eps^2}{\ell^2}  \sup_{\POVM\in {\povmset}}{\sum_{i=1}^\ell \vadj{V_i} \Choi_{\POVM} \vvec{V_i}} +16\ns c^2\eps^2 \Pr[\ptb \notin \mathcal{G}].
\label{equ:avg-MI-partial}
    \end{align}
\end{restatable}
As we have the information-theoretic limit of learning $z$, it then suffices to show that any optimal tomography estimator requires some constant amount of information. We must establish the mutual information lower bound with $\ell < d^2$, in which its not clear that there will be sufficient dependence between $z_i$ and the measurement outcomes.
\begin{lemma}
\label{lem:avg-MI-lower}
    Let $\sigma_\ptb\sim\ptbDistr(\hbasis)$ where $\ptb\sim\{-1,1\}^{\ell}$, $\out^\ns$ be the outcomes after applying $\POVM^\ns$ to $\sigma_\ptb^{\otimes\ns}$, and $\qest$ be an estimator using $\out^\ns$ that achieves an accuracy of $\eps$. Then,
    \begin{equation}
        \frac{1}{\ell}\sum_{i=1}^{\ell}\mutualinfo{z_i}{x^n}\ge\frac{1}{100}.
    \end{equation}
\end{lemma}
With these statements, we can bound the quantity associated with the restricted channel's MIC and re-arrange both sides to get a lower bound on $n$.

\subsection{Hamming separation and MI lower bound}
An important part of the MI lower bound is establishing the relationship between the trace norm between two parametrized instances and the hamming distance between the associated parameters. This allows us to relate tomography to estimating binary vectors. As $\ell$ is much lower than it was before, we need to ensure that all $z_i$ parameters still need to be estimated to estimate $\sigma_{\zest}$.
   \begin{lemma}[Trace distance Hamming separation] \label{lemma:hamm-packing}
       Consider $z \in \mathcal{G}$. For any  $\hat{z} \in \left\{-1,1\right\}^{\ell}$,
       \begin{equation}
           \tracenorm{\sigma_\ptb - \sigma_{\zest}} \geq \frac{c \eps}{2 C \ell} \cdot\ham{\ptb}{\zest}.
       \end{equation}
   \end{lemma}
\begin{proof}
It can be shown that reducing $\ell$ still yields similar results, as the operator norm will also reduce as well, from \cref{Hard}. For convenience, we will define

\begin{align*}
   C_{z} := \min\left\{1, \frac{1}{2\dims \opnorm{\Delta_{z}}}\right\} , \;
   \Delta_{w} := \frac{c \eps}{\sqrt{d\ell}}  \sum_{i=1}^\ell \indic{z_i \neq \hat{z}_i} z_i V_i, \;\Delta_{c} := \frac{c \eps}{\sqrt{d\ell}} \sum_{i=1}^\ell \indic{z_i = \hat{z}_i} z_i V_i.
\end{align*}
Notice that $C_z = 1$ since $\opnorm{W_z} \leq C \sqrt{\frac{\ell}{\dims}} \leq C \sqrt{\dims}$. By trace norm duality,
\begin{align*}
\tracenorm{\sigma_{\zest} - \sigma_{\ptb}} &=
     \tracenorm{((1+C_{\zest}) \Delta_w + (1-C_{\zest}) \Delta_c} \\
     &= \sup_{\opnorm{B} \leq 1} |\Tr[B^{\dagger} \left[(1+C_{\zest}) \Delta_w + (1-C_{\zest}) \Delta_c\right]]| \\
     &\geq \frac{\sqrt{\dims}}{C \sqrt{\ell}} \cdot
     |\Tr[W_z^{\dagger} \left[(1+C_{\zest}) \Delta_w + (1-C_{\zest}) \Delta_c\right]]| \\
     &= \frac{c \eps}{\sqrt{\dims \ell}}\cdot \frac{\sqrt{\dims}}{C \sqrt{\ell}} \cdot|(1+C_{\zest}) \delta_w + (1-C_{\zest}) \delta_c| \\
     &\geq  \frac{c \eps}{C \ell}\cdot |(1+C_{\zest}) \delta_w + (1-C_{\zest}) \delta_c| \\
     &\geq \frac{c \eps}{C \ell}\cdot \delta_w,
\end{align*}
where  $\delta_w = \ham{z}{\zest}$ and $\delta_c = \ell - \delta_w = \ell - \ham{z}{\zest}$. The reduction from $\Tr[W_z^\dagger \Delta_w] = \delta_w, \; \Tr[W_z^\dagger \Delta_c] = \delta_c$ occurs from orthormality of $V_i$.
\end{proof}
Now that we have established the trace distance Hamming separation, and can prove \cref{lem:avg-MI-lower}. The proof will proceed similarly to \cite{ADLY2025Paulinot}. 
\begin{proof}[Proof of \cref{lem:avg-MI-lower}] 
Given a procedure that computes $\hat{\sigma}$ such that $\tracenorm{\hat{\sigma}-\sigma_z} \leq \eps$ w.p. $\geq \frac{99}{100}$, we can say that the closest parametrized error to $\hat{\sigma}$, $\sigma_{\hat{z}}$, is close to $\sigma_z$ with high probability. Through a triangle inequality argument, we can say that $\tracenorm{\hat{\sigma}-\sigma_z} \leq \eps  \implies \tracenorm{\sigma_{\hat{z}}-\sigma_z} \leq 2\eps$. This results in $\Pr[\tracenorm{\sigma_{\hat{z}}-\sigma_z} \leq 2\eps] \geq \frac{99}{100}$. We decompose the expectation on this event to get
\begin{align*}
   \frac{1}{\ell}\cdot \expectDistrOf{}{\delta_w} &= \frac{1}{\ell} \cdot\expectDistrOf{}{\indic{\tracenorm{\sigma_z - \sigma_{\zest}} \leq 2 \eps} \cdot \delta_w} + \frac{1}{\ell}\cdot \expectDistrOf{}{\indic{\tracenorm{\sigma_z - \sigma_{\zest}} > 2 \eps} \cdot \delta_w} \\
   &\leq \frac{1}{\ell} \cdot\expectDistrOf{}{\indic{\tracenorm{\sigma_z - \sigma_{\zest}} \leq 2 \eps} \cdot \delta_w} + 1 \cdot \Pr[\tracenorm{\hat{\sigma}-\sigma_z} > 2\eps] \\
   &\leq \frac{1}{\ell} \cdot\expectDistrOf{}{\indic{\tracenorm{\sigma_z - \sigma_{\zest}} \leq 2 \eps} \cdot \delta_w} + 0.01.
   \label{eq:cond-expect-tom-lower}
\end{align*}
The second line uses the fact that $\tracenorm{\sigma_z -\sigma_{\hat{z}}} \leq 1$ for all $\hat{z}$. In the remaining expectation, when $z \in \mathcal{G}$,
\begin{align*}
   \frac{c \eps}{2 C \ell} \cdot \delta_w \leq \tracenorm{\sigma_z - \sigma_{\zest}} \leq 2 \eps 
   \implies  \frac{1}{\ell} \cdot \delta_w \leq \frac{4 C}{c}.
\end{align*}
Now, we use the concentration from \cref{Hard} and condition on the cases $z \in \mathcal{G}$ or $z \notin \mathcal{G}$ to get
\begin{align*}
\frac{1}{\ell} \cdot \expectDistrOf{}{\indic{\tracenorm{\sigma_z - \sigma_{\zest}} \leq 2 \eps} \cdot \delta_w} 
&\leq \Pr[\tracenorm{\sigma_z - \sigma_{\zest}} \leq 2 \eps]\left[ \Pr[z \in G] \cdot \frac{4C}{c} + \Pr[z \notin G] \cdot 1 \right] \\
&\leq \frac{4C}{c} + \Pr[z \notin G] \cdot 1
\end{align*}
$\Pr[z \notin \mathcal{G}]$ is exponentially decreasing in $\ell^{1/4}$. So for large enough $d$ and $c$, we have that the expected hamming distance is bounded by some constant.
\begin{equation}
    \label{equ:average-error-prob}
    \frac{1}{\ell} \sum_{i=1}^\ell \Pr[z_i \neq \hat{z}_i] \leq 0.41.
\end{equation}

We then apply the following lemma,
\begin{lemma}[{\cite[Lemma 10]{ACLST22iiuic}}]
\label{lem:MI-lower}
    Let $Z\in\{-1, 1\}^\ab$ be drawn uniformly and $Z-Y-\hat{Z}$ be a Markov chain where $\hat{Z}$ is an estimate of $Z$. Let $h(t)\eqdef -t\log t-(1-t)\log(1-t)$, then for each $i\in[\ab]$,
    \[
    \mutualinfo{Z_i}{Y}\ge 1-h(\probaOf{Z_i\ne \hat{Z}_i}).
    \]
\end{lemma}
Since $h(p)$ is increasing for $p \leq 1/2$, we can apply $h(\cdot)$ to both sides of \cref{equ:average-error-prob},
\begin{align*}
h(0.41) &\ge  h\left(\frac{1}{\ell} \sum_{i=1}^\ell \Pr[z_i \neq \hat{z}_i]\right)  \\
\frac{1}{\ell} \sum_{i=1}^\ell\mutualinfo{z_i}{x^n} &\geq 1 - \frac{1}{\ell} \sum_{i=1}^\ell h(\Pr[z_i \neq \hat{z}_i]) \geq 1 - h\left(\frac{1}{\ell} \sum_{i=1}^\ell \Pr[z_i \neq \hat{z}_i]\right) \geq \frac{1}{100},
\end{align*}
where the second-to-last inequality used concavity of $h(\cdot)$.
\end{proof}
\subsection{Mutual information upper bound for single qubit measurements}\label{sec:MI-single-qubit}
Recall \cref{thm:avg-MI-upper}. The key to establishing the copy complexity is to bound the spectral quantity associated with the MIC. We will show that single-qubit measurements have the same spectral bound as Pauli basis measurements. One way to look at it is that the uniform sum of Pauli matrices form a valid 2-design, which is the uniform expectation over all rotations of a single-qubit. Thus, Pauli measurements capture the complete complexity of single-qubit measurements, in expectation. \paragraph{Single qubit measurement parameterization.}
We can parametrize a single-qubit POVM observable by the Pauli decomposition of a quantum state, that is scaled down by a convex factor $\alpha_o$
\begin{align*}
    &M_o^i = \alpha_o^i (I + \beta_X^{i, o} X + \beta_Y^{i, o} Y + \beta_Z^{i, o} Z) \\
    &\sum_{o \in O} \alpha_o^i = 1,\quad {\sum_{o \in O} {\alpha_o^i\beta_\sigma^{i, o} }}=0,\\
    &{\sum_{\sigma \in \left\{X,Y,Z\right\}} {\mid\beta_\sigma^{i, o} \mid}^2} \leq 1.
\end{align*}
The conditions in the second and third line are because the POVM is required to be summed to $\eye_2$ (convex constraint) and each observable is p.s.d (norm constraint). In this form, any p.s.d and Hermitian matrix is represented with an arbitrary trace between 0 and 2, $0 \leq \alpha_o^i \leq 1$. 

\paragraph{Application to lower bound framework.}
We utilize the same strategy as providing lower bounds for Pauli basis measurements by establishing the following lemma,
\begin{lemma}\label{lemma:single_qubit_MI}
Let $V_1, ..., V_{d^2-1}$ be normalized Pauli observables ordered in decreasing weight. $\ell$ is the number of Pauli observables with weight at least $w$. Then,
\begin{align*}
    \sum_{i=1}^{\ell}\vadj{V_i} {\Choi}_\POVM \vvec{V_i} \leq \sum_{m=w}^{N}{\nqubits\choose m}.
\end{align*}
\end{lemma}
\begin{proof}
    
 We will establish the following for the general single-qubit POVM observables. The MIC of single-qubit measurements is of the following form,
\begin{align*}
    \Choi_{\POVM^i} &= \sum_{o \in O} \frac{\vvec{M_o^i} \vadj{M_o^i}}{\Tr[M_o^i]} = \sum_{o \in O} \alpha_o^i \sum_{P, P^{'} \in \Sigma \cup \{\frac{\eye}{\sqrt{2}}\}} \beta_P^{i,o} \beta_{P^{'}}^{i,o} \; \vvec{P} \vadj{{P^{'}}} \\
    \Choi_{\POVM} &= \bigotimes_{i=1}^{\nqubits} \Choi_{\POVM_i}.
\end{align*}
 For the analysis, we will be utilizing the following set to describe observables for all  $S \subseteq [N]$, 
 \begin{align*}
 \Sigma_S := \left\{\sigma \in \Sigma \cup \{\frac{\eye}{\sqrt{2}}\}^{\otimes n}\mid \sigma_{s \notin S} = \frac{\eye}{\sqrt{2}}, \sigma_{s \in S} \in \Sigma \right\}. \\
 \end{align*}
 Summing over the inner product over weight $m$ normalized Pauli observables and assuming that we have $ \beta_{\eye}^{i,o} = 1$ for all $ (i, o)\in [N]\times O$, 
\begin{align*}
  \sum_{S \in \binom{[N]}{m}} \sum_{\sigma \in \Sigma_{S}} \vadj{\sigma} {\Choi}_\POVM \vvec{\sigma} &= \sum_{S \in \binom{[N]}{m}}  \sum_{\sigma \in \Sigma_S} \prod_{i=1}^N \sum_{o \in O} \alpha_o^i \sum_{P, P^{'} \in \Sigma \cup \{\frac{\eye}{\sqrt{2}}\}} \beta_P^{i,o} \beta_{P^{'}}^{i,o}  \vvdotprod{\sigma_i}{P} \vvdotprod{P^{'}}{\sigma_i}\\
   &= \sum_{S \in \binom{[N]}{m}}  \sum_{\sigma \in \Sigma_S} \prod_{i=1}^N \sum_{o \in O} \alpha_o^i {\beta_{\sigma_i}^{i,o}}^2 \\
   &= \sum_{S \in \binom{[N]}{m}}  \sum_{\sigma \in \Sigma_S} \prod_{i \in S} \sum_{o \in O} \alpha_o^i {\beta_{\sigma_i}^{i,o}}^2 \\
   &= \sum_{S \in \binom{[N]}{m}}  \prod_{i \in S} \sum_{\sigma_i \in \Sigma}  \sum_{o \in O} \alpha_o^i {\beta_{\sigma_i}^{i,o}}^2 \\
   &= \sum_{S \in \binom{[N]}{m}}  \prod_{i \in S}  \sum_{o \in O} \alpha_o^i \sum_{\sigma_i \in \Sigma} {\beta_{\sigma_i}^{i,o}}^2 \\
   &\leq \sum_{S \in \binom{[N]}{m}}  \prod_{i \in S} 1 = \binom{N}{m}.
\end{align*}
The concluding step is due to the convexity of the coefficients $\left\{\alpha_o\right\}_{o \in O}$ and the sub-normalization of the factors $\left\{\mid\beta_{\sigma}^{i,o}\mid\right\}_{\sigma \in \Sigma}$. Thus, we get the desired bound by summing over $m=\lceil9\nqubits/10\rceil$ to $\nqubits$,
\begin{align*}
    \sum_{i=1}^{\ell}\vadj{V_i} {\Choi}_\POVM \vvec{V_i} \leq \sum_{m=9N/10}^{N}{\nqubits\choose m}.
\end{align*}
\end{proof}
\subsection{Putting things together}
When discussing lower bounds for single-qubit measurements, a combinatorial argument naturally arises to bound the spectral norm of the MIC.
Thus, it suffices to establish the combinatorial bound to guarantee almost-tight copy-complexity bounds for Pauli measurements and general single-qubit measurements. In addition to the single-qubit generalization, we have made a harder instance construction by reducing $\ell = o(d^2)$. This allows the Pauli perturbations to have a higher minimum weight, making the state harder to learn. 
\paragraph{Valid lower bound construction.} We claim that all the Pauli observables of weight at least $\lceil 9N/10 \rceil$ are a valid lower bound construction as described by \cref{sec:hard-case}.
The number of perturbation directions is
\begin{align*}
    \ell=\sum_{m=0}^{\lfloor N/10 \rfloor}\binom{\nqubits}{m}3^{\nqubits-m}\ge 3^{\lceil 9N/10 \rceil}\sum_{m=0}^{\lfloor N/10 \rfloor}\binom{\nqubits}{m}.
\end{align*}
Using Stirling's approximation on the sum of binomial coefficients \cite[Lemma 16.19]{downey2012parameterized},
\begin{align*}
    \frac{1}{\sqrt{\nqubits}}2^{\nqubits h(\frac{1}{10})}\le \sum_{m=0}^{\lfloor N/10 \rfloor}{\nqubits\choose m} \leq 2^{N h(\frac{1}{10})}.
\end{align*}
Thus,
\begin{align*}
    \ell \geq \frac{3^{\lceil 9N/10 \rceil} 2^{\nqubits h(\frac{1}{10})}}{\sqrt{N}}  \geq \frac{2^{1.88 N}}{\sqrt{N}}   
    \geq 2^{-0.33N}2^{1.88N} \geq \dims^{\frac{3}{2}}.
\end{align*}
We have shown that $\ell$ is large enough for the construction in \cref{sec:hard-case}. We continue with $w = \lceil 9N/10 \rceil$.
\paragraph{Spectral upper bound.}
Combining \cref{thm:avg-MI-upper} and \cref{lem:avg-MI-lower},
\begin{align*}
    \frac{1}{100}&\le \frac{8 \ns c^2 \eps^2}{\ell^2} \sup_{P\in\Sigma^{\otimes\nqubits}}\sum_{i=1}^{\ell}\vadj{V_i} {\Choi}_P \vvec{V_i} +16 c^2\eps^2 \Pr[\ptb \notin \mathcal{G}] \\
    &=8nc^2\eps^2\Paren{\frac{\sum_{m=0}^{\lfloor N/10 \rfloor}{\nqubits\choose m}}{\ell^2}+2\Pr[\ptb \notin \mathcal{G}]} \\
    & \le 8nc^2\eps^2\Paren{\frac{2}{9^{\lceil 9N/10 \rceil} \sum_{m=0}^{\lfloor N/10 \rfloor} \binom{\nqubits}{m}}+2\Pr[z \notin \mathcal{G}]} \\
    & \le 16nc^2\eps^2\Paren{\frac{\sqrt{N}}{9^{0.9N} 2^{Nh(\frac{1}{10})}}+\Pr[\ptb \notin \mathcal{G}]} \\
    &= 16nc^2\eps^2\Paren{\frac{\sqrt{N}}{9^{0.9N} 2^{0.1N\log 10 + 0.9N\log 10/9}}+\Pr[\ptb \notin \mathcal{G}]} \\
    &= 16nc^2\eps^2\Paren{\frac{\sqrt{N}}{9^{0.9N} 10^{0.1N} (10/9)^{0.9N}}+\Pr[\ptb \notin \mathcal{G}]} \\
    &= 16nc^2\eps^2\Paren{\frac{\sqrt{N}}{10^{N}}+\Pr[\ptb \notin \mathcal{G}]}.
\end{align*}
Note that $\probaOf{\ptb\notin \mathcal{G}}\le \dims \exp(-\ell ^{1/4})$, which is much smaller than the first term for $\dims \ge 1024$. Thus, 
\begin{align*}
    \ns=\bigOmega{\frac{10^\nqubits}{\sqrt{\nqubits}\eps^2}}.
\end{align*}

\bibliographystyle{alpha}  
\bibliography{ref}
\end{document}